\documentclass{article}

\usepackage[margin=0.5in]{geometry}

\usepackage[utf8]{inputenc} % allow utf-8 input
\usepackage[T1]{fontenc}    % use 8-bit T1 fonts
\usepackage{hyperref}       % hyperlinks
\usepackage{url}            % simple URL typesetting
\usepackage{booktabs}       % professional-quality tables
\usepackage{amsfonts,amssymb,amsmath,amsthm}
\usepackage{nicefrac}       % compact symbols for 1/2, etc.
\usepackage{microtype}      % microtypography
\usepackage{graphicx}
\usepackage{epsfig}
\usepackage{subcaption}
\usepackage{color}
\usepackage{todonotes}
\usepackage{IEEEtrantools}
\usepackage{natbib}

\graphicspath{{figures/}}
\DeclareGraphicsExtensions{.pdf,.jpeg,.png,.eps}

\newcommand{\R}{\mathbb R}
\newcommand{\N}{\mathcal N}
\newcommand{\C}{\mathcal C}
\newcommand{\E}{\mathcal E}
\newcommand{\secref}[1]{Sec.~\ref{#1}}

\newcommand{\To}{\longrightarrow}

\def\Vec#1{\!\!\hbox{$#1$\kern-0.38em\lower0.85em\hbox{$\vec{}\,$}}\,}%
\newcommand{\bbm}{\begin{bmatrix}}
\newcommand{\ebm}{\end{bmatrix}}
\newcommand{\bpm}{\begin{pmatrix}}
\newcommand{\epm}{\end{pmatrix}}

\newcommand{\mbb}[1]{{\mathbb{#1}}}
\newcommand{\mc}[1]{\mathcal{#1}}

\newcommand{\expect}[1]{\mbb{E}\left[#1\right]}

\newtheorem{theorem}{Theorem}

\title{Fully Decentralized Policies for Multi-Agent Systems: \\ An Information Theoretic Approach}
\author{
  Roel Dobbe\footnotemark[1],  David Fridovich-Keil\thanks{Indicates equal contribution.}  and Claire Tomlin 
  \thanks{Roel Dobbe, David Fridovich-Keil and Claire Tomlin are with the Department of Electrical Engineering and Computer Sciences, University of California, Berkeley, CA 94720, USA.
        {\tt\small [dobbe, dfk, tomlin]@eecs.berkeley.edu}.}%
}

\begin{document}
% \nipsfinalcopy is no longer used

\maketitle

\begin{abstract}
Learning cooperative policies for multi-agent systems is often challenged by partial observability and a lack of coordination.
In some settings, the structure of a problem allows a distributed solution with limited communication.
Here, we consider a scenario where no communication is available, and instead we learn local policies for all agents that collectively mimic the solution to a centralized multi-agent static optimization problem.
Our main contribution is an information theoretic framework based on rate distortion theory which facilitates analysis of how well the resulting fully decentralized policies are able to reconstruct the optimal solution. Moreover, this framework provides a natural extension that addresses which nodes an agent should communicate with to improve the performance of its individual policy.

\end{abstract}

\section{Introduction}
\label{sec:introduction}

%\comment{DFK: I think it's best to cast this whole thing as a static optimization problem which we're decentralizing, rather than a control problem with time dependencies and so forth. Of course we are interested in considering time dependencies, but right now our results are for static cases. Let's leave this for future work!}

% Distributed optimization arises in a variety of contexts, including model predictive control \citep{camponogara_distributed_2002}, flight formation \citep{raffard_distributed_2004} and optimal power flow in electric grids \citep{sun_fully_2013,dallanese_distributed_2013}. Broadly, these problems involve achieving some degree of coordination among agents that are solving a shared optimization problem with locally relevant variables, costs, and constraints. 
% Distributed problems can have various levels of decentralization in computation, from hybrid structures where nodes partly rely on a central or subarea nodes to fully decentralized computation where agents communication locally and communicate with other agents directly. 
% The type of coordination that can be realized is generally determined by the information structure; that is, the connectivity and capacity of the interagent communication network~\citep{camponogara_distributed_2002}.

Finding optimal decentralized policies for multiple agents is often a hard problem hampered by partial observability and a lack of coordination between agents.
The distributed multi-agent problem has been approached from a variety of angles, including distributed optimization \citep{boyd_distributed_2011}, game theory \citep{aumann_cooperative_1974} and decentralized or networked partially observable Markov decision processes (POMDPs) \citep{oliehoek_decpomdps_2016, goldman_decentralized_2004, nair_networked_2005}.
% The first approach leverages structure in the problem to form a separable problem, either via primal or dual decomposition. 
% via the Lagrangian and apply dual decomposition~\citep{raffard_distributed_2004} or the alternating direction method of multipliers (ADMM) \citep{sun_fully_2013}. 
% The second approach is to revert to a partially observed Markov decision process (POMDP) formulation that exploits network interactions between agents~\citep{nair_networked_2005}.
%Most of these methods tend to provide complex solutions that rely on some amount of communication.
% or that models the correlation between a global reward signal and the agent's individual learning signal~\citep{chang_all_2003}. 
%However, it is common for real multi-agent systems to lack reliable or fast enough communication links. 
%In such circumstances, solving a common optimization problem or implementing a cooperative strategy in a distributed fashion may be challenging or even impossible.
In this paper, we analyze a different approach consisting of a simple learning scheme to design fully decentralized policies for all agents that collectively mimic the solution to a common optimization problem, while having no access to a global reward signal and either no or restricted access to other agents' local state. This algorithm is a generalization of that proposed in our prior work \citep{sondermeijer_regression-based_2016} related to decentralized optimal power flow (OPF). Indeed, the success of regression-based decentralization in the OPF domain motivated us to understand when and how well the method works in a more general decentralized optimal control setting.

The key contribution of this work is to view decentralization as a \emph{compression} problem, and then apply classical results from information theory to analyze performance limits.
More specifically, we treat the $i^{\text{th}}$ agent's optimal action in the centralized problem as a random variable $u_i^*$, and model its conditional dependence on the global state variables $x = (x_1,\hdots,x_n)$,
% and the optimal actions $u_{j\neq i}^*$ of other agents, 
i.e. $p(u^*_i | x )$, which we assume to be stationary in time. We now restrict each agent $i$ to observe only the $i^{\text{th}}$ state variable $x_i$. Rather than solving this decentralized problem directly, we train each agent to replicate what it would have done with full information in the centralized case. That is, the vector of state variables $x$ is \textit{compressed}, and the $i^{\text{th}}$ agent must decompress $x_i$ to compute some estimate $\hat u_i \approx u_i^*$. 
In our approach, each agent learns a parameterized Markov control policy $\hat u_i = \hat \pi_i(x_i)$ via regression.
The $\hat \pi_i$ are learned from a data set containing local states $x_i$ taken from historical measurements of system state $x$ and corresponding optimal actions $u_i^*$ computed by solving an offline centralized optimization problem for each $x$.
% By ensuring that the training set is a good representation of the situations that arise in the environment, and ensuring that each local regressor well-approximates the optimal control, the collective action of all regression-based individual policies can mimic the coordinated, optimal policies. 

In this context, we analyze the fundamental limits of compression. In particular, we are interested in unraveling the relationship between the dependence structure of $u_i^*$ and $x$ and the corresponding ability of an agent with partial information to approximate the optimal solution, i.e. the difference -- or \textit{distortion} -- between decentralized action $\hat u_i = \hat \pi_i(x_i)$ and $u_i^*$.
This type of relationship is well studied within the information theory literature as an instance of \emph{rate distortion theory}~\citep[Chapter 13]{cover_elements_2012}. Classical results in this field provide a means of finding a lower bound on the expected distortion as a function of the mutual information -- or \textit{rate} of communication -- between $u_i^*$ and $x_i$. 
%Thus, given a training set $\{x[t], u^*[t]\}_{t=1}^{T}$ with $T$ data points, we can estimate the conditional distribution $p(u^* | x)$ between the full state $x$ and the optimal action $u^*$. 
%The lower bound on expected distortion then reveals the minimum average distortion between any \emph{reconstruction} $\hat u$ of $u^*$ and $u^*$ itself. 
This lower bound is valid for each specified distortion metric, and for \textit{any} arbitrary strategy of computing $\hat u_i$ from available data $x_i$.
Moreover, we are able to leverage a similar result to provide a conceptually simple algorithm for choosing a communication structure -- letting the regressor $\hat \pi_i$ depend on some other local states $x_{j \ne i}$ -- in such a way that the lower bound on expected distortion is minimized. 
As such, our method generalizes \citep{sondermeijer_regression-based_2016} and provides a novel approach for the design and analysis of regression-based decentralized optimal policies for general multi-agent systems. 
We demonstrate these results on synthetic examples, and on a real example drawn from solving OPF in electrical distribution grids. %that applied regression to decentralize the regulation of power flow and voltage in electricity grids via multiple distributed energy resources (DERs). The correlation between each DER's optimal action and the local power that each DER can sense, facilitates a decentralization of the OPF problem by learning individual mappings from local power measurement to a reactive power output of the DER that is near-optimal in terms of OPF.

\section{Related Work}
\label{sec:related_work}

% \comment{DFK: I think we should focus on the two areas I mentioned above: structured decompositions and POMDPs. I'm going to comment out all the bullets below that don't really seem that related to me.}

%Solving an optimization problem that has a common objective function with multiple cooperative agents in a distributed fashion has been approached in the past under varying conditions and using different techniques. We categorize relevant methods along two dimensions: (a) whether they explicitly solve an optimization problem or formulate a stochastic policy, and (b) whether they derive a stationary policy that is independent of time or a non-stationary policy that represents a sequence of action mappings indexed by time. \comment{I think we should emphasize these dimensions in the bullet points below, if we want to bring them up at all.}
Decentralized control has long been studied within the system theory literature, e.g. \citep{lunze_feedback_1992, siljak_complex_2011}. 
Recently, various decomposition based techniques have been proposed for distributed optimization based on primal or dual decomposition methods, which all require iterative computation and some form of communication with either a central node \citep{boyd_distributed_2011} or neighbor-to-neighbor on a connected graph \citep{pu_inexact_2014,raffard_distributed_2004,sun_fully_2013}.
Distributed model predictive control (MPC) optimizes a networked system composed of subsystems over a time horizon, which can be decentralized (no communication) if the dynamic interconnections between subsystems are weak in order to achieve closed-loop stability as well as performance~\citep{christofides_distributed_2013}. 
The work of \citet{zeilinger_plug_2013} extended this to systems with strong coupling by employing time-varying distributed terminal set constraints, which requires neighbor-to-neighbor communication. 
Another class of methods model problems in which agents try to cooperate on a common objective without full state information as a decentralized partially observable Markov decision process (Dec-POMDP)~\citep{oliehoek_decpomdps_2016}. \citet{nair_networked_2005} introduce networked distributed POMDPs, a variant of the Dec-POMDP inspired in part by the pairwise interaction paradigm of distributed constraint optimization problems (DCOPs).
% In principle, this framework can be used for both designing stationary or non-stationary policies. 

%An additional difficulty in the multi-agent POMDP setting occurs when agents have to learn online from a common global reward signal, which effectively masks which agents bear responsibility for changes in reward~\citep{chang_all_2003}.
% proposed artificially modelling the observed global reward signal as the sum of an agent's own contribution and a random Markov process accounting for the amount of the observed reward due to other agents or external factors.

Although the specific algorithms in these works differ significantly from the regression-based decentralization scheme we consider in this paper, a larger difference is in problem formulation. As described in Sec. \ref{sec:problem_formulation}, we study a static optimization problem repeatedly solved at each time step. Much prior work, especially in optimal control (e.g. MPC) and reinforcement learning (e.g. Dec-POMDPs), poses the problem in a dynamic setting where the goal is to minimize cost over some time horizon. In the context of reinforcement learning (RL), the time horizon can be very long, leading to the well known tradeoff between exploration and exploitation; this does not appear in the static case. Additionally, many existing methods for the dynamic setting require an ongoing communication strategy between agents -- though not all, e.g. \citep{peshkin_learning_2000}. Even one-shot static problems such as DCOPs tend to require complex communication strategies, e.g. \citep{modi_adopt_2005}.

%The innovative aspect of our approach is to learn \emph{decentralized} stationary policies for all agents based on historical measurements. These policies can be used at run-time \emph{without} communication and collectively mimic the solution to a multi-agent optimization problem.
%In addition, in contrast to most popular distributed optimization methods, this technique does not need any iteration at run-time to converge to a feasible and optimal solution, and can serve problems or dynamics that require a fast response time. 
% (such as voltage variability in electric grids).
% agent in on offline setting only once, and then apply that policy at run-time. %(and can be updated when needed).
% Each individual policy is designed as a regression which maps the compressed state information to an action. The regression will be trained on a dataset where the state is taken from historical measurements of the system and the corresponding optimal action at each node, where that action is computed by solving the full centralized optimization problem.
% By ensuring that the training set is a good representation of the situations that arise in the environment, and ensuring that each local regressor well-approximates the optimal action, the collective action of all regression-based individual policies can mimic the coordinated, optimal policies. 

Although the mathematical formulation of our approach is rather different from prior work, the policies we compute are similar in spirit to other learning and robotic techniques that have been proposed, such as behavioral cloning \citep{sammut_automatic_1996} and apprenticeship learning \citep{abbeel_apprenticeship_2004}, which aim to let an agent learn from examples.
In addition, we see a parallel with recent work on information-theoretic bounded rationality \citep{ortega_information-theoretic_2015} which seeks to formalize decision-making with limited resources such as the time, energy, memory, and computational effort allocated for arriving at a decision. 
Our work is also related to swarm robotics \citep{brambilla_swarm_2013}, as it learns simple rules aimed to design robust, scalable and flexible collective behaviors for coordinating a large number of agents or robots.

\section{General Problem Formulation}
\label{sec:problem_formulation}
% \todo[inline]{Do we need ``dynamics'' here? I mean, it's obviously important in general but it only appears in the constraints, right? So can't we just forget about it here?}

% \todo[inline]{We may want to distinguish between states that we can measure and are not dynamically coupled, and states that are (e.g. nodal power vs voltage). (To assume that actions can be characterized and modeled as a reaction to the non-dynamical state).}

Consider a distributed multi-agent problem defined by a graph $\mathcal G = (\N, \E)$, with $\N$ denoting the nodes in the network with cardinality $|\N| = N$, and $\E$ representing the set of edges between nodes. Fig. \ref{fig:network} shows a prototypical graph of this sort.
Each node has a real-valued state vector $x_i \in \R^{\alpha_i}\,,i\in~\N$. A subset of nodes $\mathcal C \subset \mathcal N$, with cardinality $|\mathcal C| = C$, are controllable and hence are termed ``agents.'' Each of these agents has an action variable~$u_i \in \R^{\beta_i}\,,i\in\C$. 
Let $x = (x_i,\hdots,x_N)^{\top} \in \R^{\sum_{i \in \mc{N}} \alpha_i} = \mc{X}$ denote the full network state vector and $u \in \R^{\sum_{i \in \mc{C}} \beta_i} = \mc{U}$ the stacked network optimization variable. 
%Adjacent nodes $(i,j) \in \E$ are dynamically coupled and physically constrained through a equations $g_{ij}(x_i,x_j,u_i,u_j) = 0$, which collectively yields equations $g(x,u) = 0$. 
Physical constraints such as spatial coupling are captured through equality constraints $g(x,u) = 0$.
In addition, the system is subject to inequality constraints $h(x,u) \le 0$ that incorporate limits due to capacity, safety, robustness, etc.
We are interested in minimizing a convex scalar function $f_o(x,u)$ that encodes objectives that are to be pursued cooperatively by all agents in the network, i.e. we want to find
\begin{equation}
\label{eqn:optimization_problem}
\begin{array}{rl}
\displaystyle u^* =  \arg \min_{u}  & \quad f_o(x,u) \,, \\
\text{s.t.} &\quad  g(x,u) = 0, \quad h(x,u) \le 0.
\end{array}
\end{equation}

Note that \eqref{eqn:optimization_problem} is static in the sense that it does not consider the future evolution of the state $x$ or the corresponding future values of cost $f_o$. We apply this static problem to sequential control tasks by repeatedly solving \eqref{eqn:optimization_problem} at each time step. Note that this simplification from an explicitly dynamic problem formulation (i.e. one in which the objective function incorporates future costs) is purely for ease of exposition and for consistency with the OPF literature as in \citep{sondermeijer_regression-based_2016}. We could also consider the optimal policy which solves a dynamic optimal control or RL problem and the decentralized learning step in Sec. \ref{subsec:decentralized_learning} would remain the same. 

Since \eqref{eqn:optimization_problem} is static, applying the learned decentralized policies repeatedly over time may lead to dynamical instability. Identifying when this will and will not occur is a key challenge in verifying the regression-based decentralization method, however it is beyond the scope of this work.

\begin{figure}[t]
\centering
\begin{subfigure}{.5\textwidth}
  \centering
  \includegraphics[height=0.13\textheight]{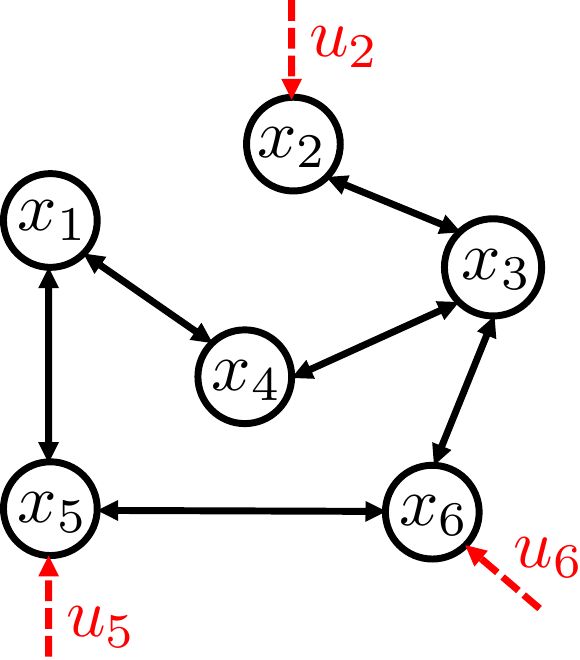}
  \caption{Distributed multi-agent problem.}
  \label{fig:network}
\end{subfigure}%
\begin{subfigure}{.5\textwidth}
  \centering
  \includegraphics[height=0.13\textheight]{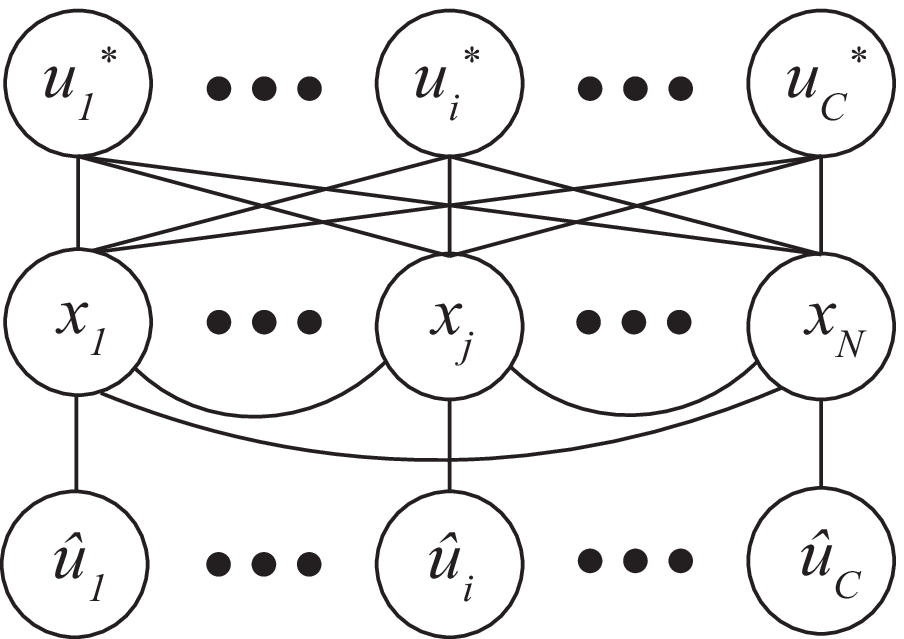}
  \caption{Graphical model of dependency structure.}
  \label{fig:graphical_model}
\end{subfigure}%
\caption{(a) shows a connected graph corresponding to a distributed multi-agent system. The circles denote the local state $x_i$ of an agent, the dashed arrow denotes its action $u_i$, and the double arrows denote the physical coupling between local state variables. (b) shows the Markov Random Field (MRF) graphical model of the dependency structure of all variables in the decentralized learning problem. Note that the state variables $x_i$ and the optimal actions $u_i^*$ form a fully connected undirected network, and the local policy $\hat u_i$ only depends on the local state $x_i$. \label{fig:overview_figs}}
\end{figure}

\subsection{Decentralized Learning}
\label{subsec:decentralized_learning}
We interpret the process of solving \eqref{eqn:optimization_problem} as applying a well-defined function or stationary Markov policy $\pi^* : \mc{X} \To \mc{U}$ that maps an input collective state $x$ to the optimal collective control or action $u^*$. We presume that this solution exists and can be computed offline.
% \comment{Is this enough, or should $\pi^*$ be formulated mathematically too?}
Our objective is to learn $C$ decentralized policies $\hat u_i = \hat \pi_i (x_i)$, one for each agent $i \in \C$, based on $T$ historical measurements of the states $\{x[t]\}_{t=1}^T$ and the offline computation of the corresponding optimal actions $\{u^*[t]\}_{t=1}^T$. 
Although each policy $\hat \pi_i$ individually aims to approximate $u^*_i$ based on local state $x_i$, we are able to reason about how well their collective action can approximate $\pi^*$.
Figure~\ref{fig:flow_diagram} summarizes the decentralized learning setup.
%Imagine we solve this problem for a set of $T$ historical readings $\{x(t)\}_{t=1}^{T}$, yielding a set of minimizers $\{u^*(t)\}_{t=1}^{T}$. We can separate the overall data set into $C$ smaller data sets $\{x_i(t), u^*_i(t)\}_{t=1}^{T} \ , \ \forall i \in \mathcal C$, with $x_i \in \mathbb R^{\alpha_x}$ and $u^*_i \in \mathbb R^{\alpha_u}$. We now consider the problem of parametrizing the policy $\hat \pi_i (x_i, \theta_i, \phi(\cdot))$, where the data $x_i$ can be transformed by a kernel fucntion $\phi_i(\cdot)$ that is potentially different for each node, and $\theta_i$ denotes the regression parameters. If we restrict the regression function to be a linear combination of a vector of features $\phi_i(x_i)$, this yields
%\begin{equation}
%\hat \pi_i (x_i) = \theta_i^{\top} \phi_i(x_i) \ , \quad \forall i \in \mathcal C \,.
%\end{equation}
%The practical challenges are selecting the best feature kernel $\phi_i(\cdot)$ and determining the associated parameters $\theta_i$ for all controlled nodes in $\mathcal C$ in a way that collectively all local policies $\hat \pi_i$ closely approximate the multi-agent policy $\pi^*$ that minimizes~\eqref{eqn:optimization_problem}. More fundamentally, we are interested in the fundamental limits on reconstructing $\pi^*$ based on local information $x_i\,, \forall i \in \C$. In addition, in the case that local reconstructability is limited, we are interested to determine with which node(s) to communicate in order to improve the reconstruction.

\begin{figure}[t]
\centering
  \includegraphics[width=0.8\textwidth]{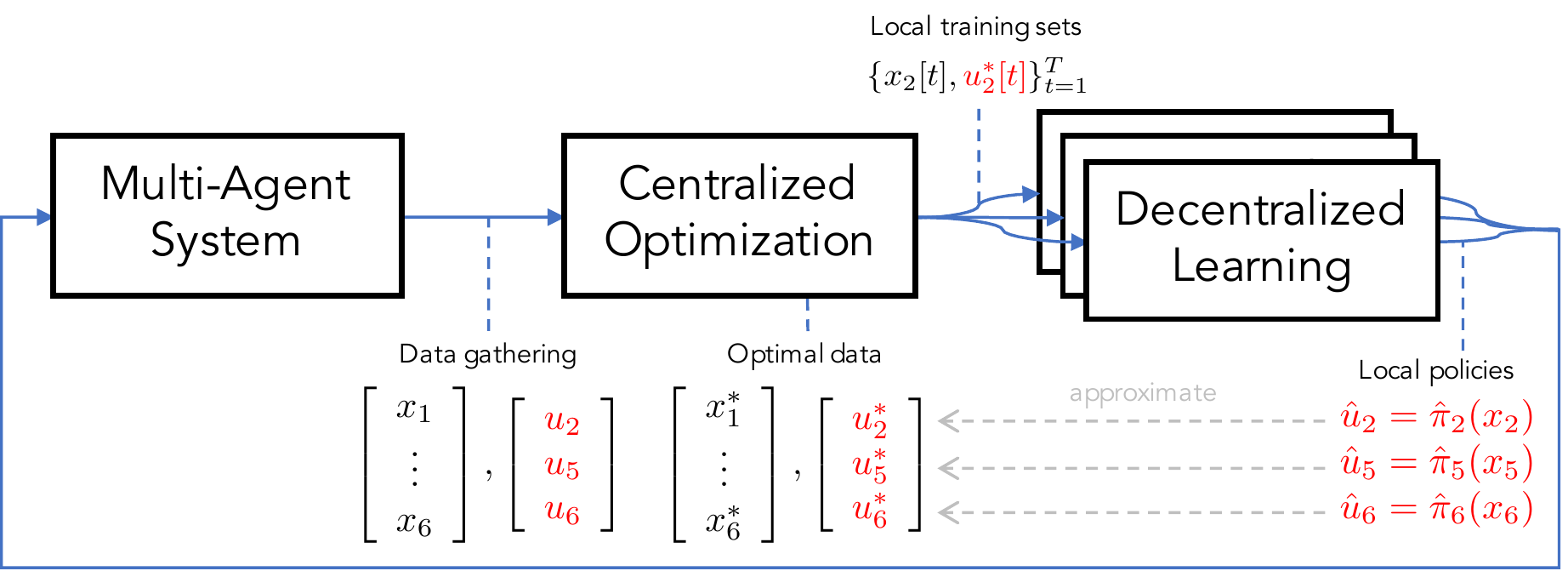}
  \caption{A flow diagram explaining the key steps of the decentralized regression method, depicted for the example system in Fig. \ref{fig:network}. We first collect data from a multi-agent system, and then solve the centralized optimization problem using all the data. The data is then split into smaller training and test sets for all agents to develop individual decentralized policies $\hat \pi_i(x_i)$ that approximate the optimal solution of the centralized problem. These policies are then implemented in the multi-agent system to collectively achieve a common global behavior.}
  \label{fig:flow_diagram}
\end{figure}

More formally, we describe the dependency structure of the individual policies $\hat \pi_i : \mathbb R^{\alpha_i} \To \mathbb R^{\beta_i}$ with a Markov Random Field (MRF) graphical model, as shown in Fig. \ref{fig:graphical_model}. The $\hat u_i$ are only allowed to depend on local state $x_i$ while the $u_i^*$ may depend on the full state $x$.
With this model, we can determine how information is distributed among different variables and what information-theoretic constraints the policies $\{\hat \pi_i\}_{i \in \C}$ are subject to when collectively trying to reconstruct the centralized policy~$\pi^*$. Note that although we may refer to $\pi^*$ as globally optimal, this is not actually required for us to reason about how closely the $\hat \pi_i$ approximate $\pi^*$. That is, our analysis holds even if \eqref{eqn:optimization_problem} is solved using approximate methods. In a dynamical reformulation of \eqref{eqn:optimization_problem}, for example, $\pi^*$ could be generated using techniques from deep RL.

\subsection{A Rate-Distortion Framework}
\label{subsec:rate_distortion_perspective}

We approach the problem of how well the decentralized policies  $\hat \pi_i$ can perform in theory from the perspective of \emph{rate distortion}.
Rate distortion theory is a sub-field of information theory which provides a framework for understanding and computing the minimal \textit{distortion} incurred by any given \textit{compression} scheme. 
In a rate distortion context, we can interpret the fact that the output of each individual policy $\hat \pi_i$ depends only on the local state $x_i$ as a compression of the full state $x$.
For a detailed overview, see~\citep[Chapter 10]{cover_elements_2012}. 
We formulate the following variant of the the classical rate distortion problem
\begin{align}
\label{eqn:rate_distortion_formulation}
D^* = \min_{p(\hat u | u^*)} \, & \quad \expect{d(\hat u, u^*)} \,, \\
\textnormal{s.t. } \, & \quad I(\hat u_i; u_j^*) \le I(x_i; u_j^*) \triangleq \gamma_{ij} \,, \nonumber\\
& \quad I(\hat u_i; \hat u_j) \le I(x_i; x_j) \triangleq \delta_{ij}, \forall i, j \in \C \,, \nonumber
\end{align}
% \begin{theorem} (Rate Distortion)
% \label{thm:rate_distortion}

% The rate distortion function $R(D)$ for a so-called ``source'' $A \in \mc{A}, \sim p_{A}\ i.i.d.$, ``reconstruction'' $\hat A \in \mc{A}$, distortion function $d : \mc{A} \times \mc{A} \To \mbb{R}_{\ge 0}$, and threshold $D$ is given by
% \begin{align}
% \label{eqn:rate_distortion}
% R(D) = \min_{p(\hat A | A) :\  \expect{d(A, \hat A)} \le D} I(A; \hat A) \,.
% \end{align}
% \end{theorem}

% Theorem \ref{thm:rate_distortion} tells us that the minimum \textit{rate}, or number of bits per symbol, required to encode symbols from some i.i.d. random variable $A$, is given by the minimum of mutual information between $A$ and $\hat A$. In the field, $\hat A$ is taken to be the output of a ``decoder'' or ``decompressor'' $\hat \pi$ which is attempting to reconstruct the original ``source'' $A$ from limited information. Note that the minimization is over conditional distributions $p(\hat A | A)$, i.e. decoding errors, that incur at most distortion $D$ according to an arbitrary metric $d$. In our case, we replace $A$ with $u_i^*$ and $\hat A$ with $\hat u_i$, for each $i \in \mc{C}$.

% Although Thm. \ref{thm:rate_distortion} is usually stated as above, it is equally valid to swap the constraint and objective to define the distortion rate function
% \begin{align}
% \label{eqn:reversed_rate_distortion}
% D(R) = \min_{p(\hat u | u^*) :\ I(u^*; \hat u) \le R} \expect{d(u^*, \hat u)}
% \end{align}
where $I(\cdot,\cdot)$ denotes mutual information and $d(\cdot,\cdot)$ an arbitrary non-negative distortion measure. As usual, the minimum distortion between random variable $u^*$ and its reconstruction $\hat u$ may be found by minimizing over conditional distributions $p(\hat u | u^*)$.

The novelty in \eqref{eqn:rate_distortion_formulation} lies in the structure of the constraints. Typically, $D^*$ is written as a function $D(R)$, where $R$ is the maximum \textit{rate} or mutual information $I(\hat u; u^*)$. From Fig. \ref{fig:graphical_model} however, we know that pairs of reconstructed and optimal actions cannot share more information than is contained in the intermediate nodes in the graphical model, e.g. $\hat u_1$ and $u_1^*$ cannot share more information than $x_1$ and $u_1^*$. This is a simple consequence of the data processing inequality \citep[Thm. 2.8.1]{cover_elements_2012}. Similarly, the reconstructed optimal actions at two different nodes cannot be more closely related than the measurements $x_i$'s from which they are computed.
The resulting constraints are fixed by the joint distribution of the state $x$ and the optimal actions $u^*$. 
That is, they are fully determined by the structure of the optimization problem \eqref{eqn:optimization_problem} that we wish to solve.

We emphasize that we have made virtually no assumptions about the distortion function. For the remainder of this paper, we will measure distortion as the deviation between $\hat u_i$ and $u_i^*$. However, we could also define it to be the suboptimality gap $f_o(x, \hat u) - f_o(x, u^*)$, which may be much more complicated to compute. This definition could allow us to reason explicitly about the cost of decentralization, and it could address the valid concern that the optimal decentralized policy may bear no resemblance to $\pi^*$. We leave further investigation for future work.

\subsection{Example: Squared Error, Jointly Gaussian}
\label{subsec:joint_gaussian_squared_error}

% \todo[inline]{Before we go into the Gaussian case below, here we can argue that, if we assume a setting in which both $u_i^*$ and $\hat u_i$ can be expressed as a linear combination of features of $x_i$  (such as in our OPF application), then the squared distortion function is minimized by a linear model, which is the linear least squares estimator, also referred to as the best linear unbiased estimator. \\
To provide more intuition into the rate distortion framework, we consider an idealized example in which the $x_i, u_i \in \mbb{R}^1$. Let  $d(\hat u, u^*) = \|\hat u - u^*\|_2^2$ be the squared error distortion measure, and assume the state $x$ and optimal actions $u^*$ to be jointly Gaussian.
These assumptions allow us to derive an explicit formula for the optimal distortion $D^*$ and corresponding regression policies $\hat \pi_i$.
We begin by stating an identity for two jointly Gaussian $X, Y \in \mbb{R}$ with correlation $\rho$:
% \begin{align*}
$ \, I(X; Y) \le \gamma \iff \rho^2 \le 1 - e^{-2 \gamma} \,$,
% \end{align*}
which follows immediately from the definition of mutual information and the formula for the entropy of a Gaussian random variable.
Taking $\rho_{\hat u_i, u_i^*}$ to be the correlation between $\hat u_i$ and $u_i^*$, $\sigma_{\hat u_i}^2$ and $\sigma_{u_i^*}^2$ to be the variances of $\hat u_i$ and $u_i^*$ respectively, and assuming that $u_i^*$ and $\hat u_i$ are of equal mean (unbiased policies $\hat \pi_i$), we can show that the minimum distortion attainable is
\begin{align}
D^* &= \min_{p(\hat u | u^*)} \expect{\|u^* - \hat u\|_2^2} : \rho_{\hat u_i, u_i^*}^2 \le 1 - e^{-2 \gamma_{ii}} = \rho_{u_i^*, x_i}^2, \forall i \in \mc{C} \,,\\
% &= \min_{\{\rho_{\hat u_i, u_i^*}\}, \{\sigma_{\hat u_i}\}} \sum_i \expect{(u_i^* - \hat u_i)^2} : \rho_{\hat u_i, u_i^*}^2 \le \rho_{u_i^*, x_i}^2 \,, \\
 &= \min_{\{\rho_{\hat u_i, u_i^*}\}, \{\sigma_{\hat u_i}\}} \sum_i \left(\sigma_{u_i^*}^2 + \sigma_{\hat u_i}^2 - 2 \rho_{\hat u_i, u_i^*} \sigma_{u_i^*} \sigma_{\hat u_i}\right) : \rho_{\hat u_i, u_i^*}^2 \le \rho_{u_i^*, x_i}^2 \,, 
 \label{eqn:solve_rho} \\
&= \min_{\{\sigma_{\hat u_i}\}} \sum_i \left(\sigma_{u_i^*}^2 + \sigma_{\hat u_i}^2 - 2 \rho_{u_i^*, x_i} \sigma_{u_i^*} \sigma_{\hat u_i}\right) \,, 
\label{eqn:solve_sigma} \\
% &= \sum_i \left(\sigma_{u_i^*}^2 (1 + \rho_{u_i^*, x_i}^2) - 2 \rho_{u_i^*, x_i}^2 \sigma_{u_i^*}^2\right) \,, \\
&= \sum_i \sigma_{u_i^*}^2 (1 - \rho_{u_i^*, x_i}^2) \,.
\label{eqn:gaussian_d*}
\end{align}
In \eqref{eqn:solve_rho}, we have solved for the optimal correlations $\rho_{\hat u_i, u_i^*}$. 
Unsurprisingly, the optimal value turns out to be the maximum allowed by the mutual information constraint, i.e. $\hat u_i$ should be as correlated to $u_i^*$ as possible, and in particular as much as $u_i^*$ is correlated to $x_i$. Similarly, in \eqref{eqn:solve_sigma} we solve for the optimal $\sigma_{\hat u_i}$, with the result that at optimum, $\sigma_{\hat u_i} = \rho_{u_i^*, x_i} \sigma_{u_i^*}$. 
This means that as the correlation between the local state $x_i$ and the optimal action $u^*_i$ decreases, the variance of the estimated action $\hat u_i$ decreases as well. As a result, the learned policy will increasingly ``bet on the mean'' or ``listen less'' to its local measurement to approximate the optimal action. 

Moreover, we may also provide a closed form expression for the regressor which achieves the minimum distortion $D^*$. 
Since we have assumed that each $u_i^*$ and the state $x$ are jointly Gaussian, we may write any $u_i^*$ as an affine function of $x_i$ plus independent Gaussian noise. Thus, the minimum mean squared estimator is given by the conditional expectation
% i.e. $u_i^* = \expect{u_i^*} + \sum_j a_j (x_j - \expect{x_j}) + w_i$, where the coefficients $a_j$ are fixed constants and $w_i$ is a zero-mean Gaussian-distributed random variable independent of all the $x$'s. The further assumption that the regression outputs $\hat u$ are also jointly Gaussian with $u^*$ and $x$ implies that $\hat u_i = \expect{u^*} + b (x_i - \expect{x_i})$ for some constant $b$. Solving for the $b$ that minimizes the expected distortion metric (mean squared error with respect to $u_i^*$) is equivalent to finding the best linear estimator for $u_i^*$ given only $x_i$, which may be solved via Gaussian conditioning
\begin{align}
\hat u_i = \hat \pi_i(x_i) &= \expect{u_i^* | x_i} = \expect{u_i^*} + \frac{\rho_{u_i^* x_i} \sigma_{u_i^*}}{\sigma_{x_i}} (x_i - \expect{x_i}) \,.
\label{eqn:linear_regressor}
\end{align}
Thus, we have found a closed form expression for the best regressor $\hat \pi_i$ to predict $u_i^*$ from only $x_i$ in the joint Gaussian case with squared error distortion. 
This result comes as a direct consequence of knowing the true parameterization of the joint distribution $p(u^*, x)$ (in this case, as a Gaussian).

\subsection{Determining Minimum Distortion in Practice}
\label{subsec:determining_min_distortion_practice}

Often in practice, we do not know the parameterization $p(u^*| x)$ and hence it may be intractable to determine $D^*$ and the corresponding decentralized policies $\hat \pi_i$.
However, if one can assume that $p(u^*| x)$ belongs to a family of parameterized functions (for instance universal function approximators such as deep neural networks), then it is theoretically possible to attain or at least approach minimum distortion for arbitrary non-negative distortion measures. 

Practically, one can compute the mutual information constraint $I(u_i^*,x_i)$ from \eqref{eqn:rate_distortion_formulation} to understand how much information a regressor $\hat \pi_i(x_i)$ has available to reconstruct $u_i^*$. 
In the Gaussian case, we were able to compute this mutual information in closed form. 
For data from general distributions however, there is often no way to compute mutual information analytically. 
Instead, we rely on access to sufficient data $\{x[t],u^*[t]\}_{t=1}^T$, in order to estimate mutual informations numerically. 
In such situations (e.g. \secref{sec:application_optimal_power_flow}), we discretize the data and then compute mutual information with a minimax risk estimator, as proposed by \citet{jiao_minimax_2014}.

\section{Allowing Restricted Communication}
\label{sec:allowing_restricted_communication}

Suppose that a decentralized policy $\hat \pi_i$ suffers from insufficient mutual information between its local measurement $x_i$ and the optimal action $u^*_i$. 
In this case, we would like to quantify the potential benefits of communicating with other nodes $j \neq i$ in order to reduce the distortion limit $D^*$ from \eqref{eqn:rate_distortion_formulation} and improve its ability to reconstruct $u_i^*$.
In this section, we present an information-theoretic solution to the problem of how to choose optimally which other data to observe, and we provide a lower bound-achieving solution for the idealized Gaussian case introduced in \secref{subsec:joint_gaussian_squared_error}.
We assume that in addition to observing its own local state $x_i$, each $\hat \pi_i$ is allowed to depend on at most $k$ other $x_{j \ne i}$. 
% Our main result is stated below in Thm.~\ref{thm:restricted_communication}.
\begin{theorem} (Restricted Communication)
\label{thm:restricted_communication}

If $\mc{S}_i$ is the set of $k$ nodes $j \ne i \in \N$ which $\hat u_i$ is allowed to observe in addition to $x_i$, then setting 
\begin{align}
\label{eqn:restricted_communication}
\mc{S}_i = \arg \max_{\mc{S}} \ I(u_i^*; x_i, \{x_j : j \in \mc{S}\}) \ : \ |\mc{S}| = k \,,
\end{align}
minimizes the best-case expectation of \textit{any} distortion measure. That is, this choice of $\mc{S}_i$ yields the smallest lower bound $D^*$ from \eqref{eqn:rate_distortion_formulation} of any possible choice of $\mc{S}$.
\end{theorem}
\begin{proof}
By assumption, $\mc{S}_i$ maximizes the mutual information between the observed local states $\{x_i,~x_j~:~j~\in~\mc{S}_i\}$ and the optimal action $u_i^*$. This mutual information is equivalent to the notion of \textit{rate}~$R$ in the classical rate distortion theorem \citep{cover_elements_2012}. It is well-known that the distortion rate function~$D(R)$ is convex and monotone decreasing in $R$. Thus, by maximizing mutual information~$R$ we are guaranteed to minimize distortion~$D(R)$, and hence $D^*$.
\end{proof}

Theorem \ref{thm:restricted_communication} provides a means of choosing a subset of the state $\{x_j : j \ne i\}$ to communicate to each decentralized policy $\hat \pi_i$ that minimizes the corresponding best expected distortion $D^*$. Practically speaking, this result may be interpreted as formalizing the following intuition: ``the best thing to do is to transmit the most information.'' In this case, ``transmitting the most information'' corresponds to allowing $\hat \pi_i$ to observe the set $\mc{S}$ of nodes $\{x_j : j \ne i\}$ which contains the most information about~$u_i^*$. Likewise, by ``best'' we mean that $\mc{S}_i$ minimizes the best-case expected distortion~$D^*$, for any distortion metric~$d$. As in \secref{subsec:joint_gaussian_squared_error}, without making some assumption about the structure of the distribution of $x$ and $u^*$, we cannot guarantee that any particular regressor $\hat \pi_i$ will attain~$D^*$.
Nevertheless, in a practical situation where sufficient data $\{x[t],u^*[t]\}_{t=1}^T$ is available, we can solve \eqref{eqn:restricted_communication} by estimating mutual information~\citep{jiao_minimax_2014}.
% 
% Note however that Thm. \ref{thm:restricted_communication} provides neither a way to compute $D^*$, nor a means of choosing a regression function -- or equivalently, a parameterization of $p(\hat u | u^*; \phi)$ -- which achieves $D^*$. Section \ref{subsec:joint_gaussian_squared_error_communication} presents one case in which we can derive an explicit expression for $D^*$ given $\mc{S}_i$ and moreover a means for achieving $D^*$. Section \ref{subsec:comparing_strategies_opf} presents a real example drawn from the power systems literature.

\subsection{Example: Joint Gaussian, Squared Error with Communication}
\label{subsec:joint_gaussian_squared_error_communication}

Here, we reexamine the joint Gaussian-distributed, mean squared error distortion case from \secref{subsec:joint_gaussian_squared_error}, and apply Thm. \ref{thm:restricted_communication}.
We will take $u^* \in \mbb{R}^1, x \in \mbb{R}^{10}$ and $u^*, x$ jointly Gaussian with zero mean and arbitrary covariance. The specific covariance matrix $\Sigma$ of the joint distribution $p(u^*, x)$ is visualized in Fig. \ref{fig:synthetic_covariance}. For simplicity, we show the squared correlation coefficients of $\Sigma$ which lie in $[0, 1]$.
The boxed cells in $\Sigma$ in Fig. \ref{fig:synthetic_covariance} indicate that $x_9$ solves \eqref{eqn:restricted_communication}, i.e. $j = 9$ maximizes $I(u^*; x_1, x_j)$ the mutual information between the observed data and regression target $u^*$. Intuitively, this choice of $j$ is best because $x_9$ is highly correlated to $u^*$ and weakly correlated to $x_1$, which is already observed by $\hat u$; that is, it conveys a significant amount of information about $u^*$ that is not already conveyed by $x_1$.

\begin{figure}[t]
\centering
\begin{subfigure}{.5\textwidth}
  \centering
  \includegraphics[height=0.18\textheight]{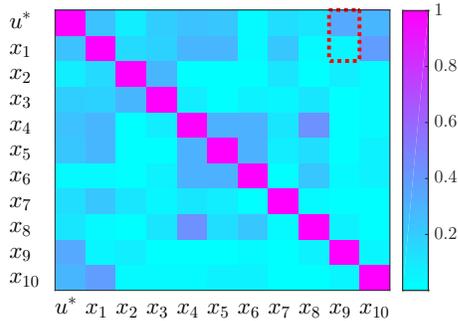}
  \caption{Squared correlation coefficients. \label{fig:synthetic_covariance}}
\end{subfigure}%
\begin{subfigure}{.5\textwidth}
  \centering
  \includegraphics[height=0.18\textheight]{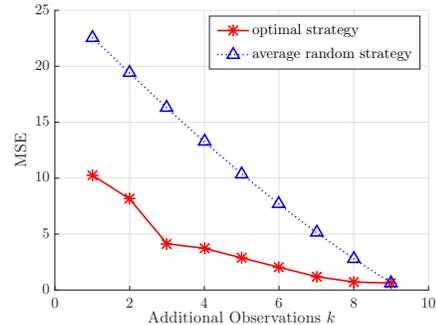}
  \caption{Comparison of communication strategies. \label{fig:comparison_communication_strategies_gaussian}}
\end{subfigure}%
\caption{Results for optimal communication strategies on a synthetic Gaussian example. (a) shows squared correlation coefficients between of $u^*$ and all $x_i$'s. The boxed entries correspond to $x_9$, which was found to be optimal for $k = 1$. (b) shows that the optimal communication strategy of Thm. \ref{thm:restricted_communication} achieves the lowest average distortion and outperforms the average over random strategies. \label{fig:results_gauss}}
\end{figure}

% Given $\mc{S}_i$ the solution to \eqref{eqn:restricted_communication}, we may derive the corresponding lower bound on distortion $D^*$ following the process in \eqref{eqn:variance_gaussian_conditioning}. For convenience, we define $z_i \triangleq (x_i, x_j : j \in \mc{S}_i)$, which yields
% \begin{align}
% D^* &= \sum_i \textnormal{Var}(u_i^* | z_i) = \sum_i \left(\sigma_{u_i^*}^2 - \Sigma_{u_i^* z_i} \Sigma_{z_i}^{-1} \Sigma_{u_i^* z_i}\right) \,,
% \label{eqn:gaussian_communication_d*}
% \end{align}
% where $\Sigma_{z_i}$ is the covariance matrix of the observed data $z_i$ and $\Sigma_{u_i^* z_i}$ is the cross covariance vector of the target $u_i^*$ and the observed data.
% Again, following \eqref{eqn:linear_regressor} we derive the optimal regression function for $u_i^*$ given observed data $z_i$, i.e.
% \begin{align}
% \hat u_i &= \expect{u_i^* | z_i} = \expect{u_i^*} + \Sigma_{u_i^* z_i} \Sigma_{z_i}^{-1} (z_i - \expect{z_i}) \,.
% \label{eqn:linear_regressor_communication}
% \end{align}
% 
Figure \ref{fig:comparison_communication_strategies_gaussian} shows empirical results. Along the horizontal axis we increase the value of $k$, the number of additional variables $x_j$ which regressor $\hat \pi_i$ observes. The vertical axis shows the resulting average distortion. We show results for a linear regressor of the form of \eqref{eqn:linear_regressor} where we have chosen $\mc{S}_i$ optimally according to \eqref{eqn:restricted_communication}, as well as uniformly at random from all possible sets of unique indices. Note that the optimal choice of $\mc{S}_i$ yields the lowest average distortion $D^*$ for all choices of $k$. Moreover, the linear regressor of \eqref{eqn:linear_regressor} achieves $D^*$ for all $k$, since we have assumed a Gaussian joint distribution.

% In the Gaussian case, there is a strong connection between entropy and variance. For a Gaussian random variable $X$, the entropy of $X$ is given by
% \begin{align}
% \label{eqn:gaussian_entropy}
% h(X) = \frac{1}{2} \log(2 \pi e \textnormal{Var}(X)) \,.
% \end{align}
% Since the mutual information between Gaussian random variables $X$ and $Y$ is given by $I(X; Y) \triangleq h(X) - h(X | Y)$, it is maximized when $h(X | Y)$ -- or equivalently Var$(X | Y)$ is minimized. 

% This suggests an interesting connection to classical Principal Component Analysis (PCA). In PCA, one seeks to find the subspace of $\mbb{R}^n$ in which the variance of some data is maximized. In Thm. \ref{thm:restricted_communication} with Gaussian data, we are doing essentially the same thing, except that we are not seeking \textit{any} subspace, but an axis-aligned subspace. Thus, solving \eqref{eqn:restricted_communication} in the Gaussian case is equivalent to PCA if the covariance matrix of all the variables is diagonal.

\section{Application to Optimal Power Flow}
\label{sec:application_optimal_power_flow}

In this case study, we aim to minimize the voltage variability in an electric grid caused by intermittent renewable energy sources and the increasing load caused by electric vehicle charging. 
We do so by controlling the reactive power output of distributed energy resources (DERs), while adhering to the physics of power flow and constraints due to energy capacity and safety.
Recently, various approaches have been proposed, such as \citep{farivar_equilibrium_2013} or \citep{zhang_optimal_2014}. 
In these methods, DERs tend to rely on an extensive communication infrastructure, either with a central master node \citep{xu_multi-timescale_2017} or between agents leveraging local computation \citep{dallanese_optimal_2014}.
We study regression-based decentralization as outlined in \secref{sec:problem_formulation} and Fig. \ref{fig:flow_diagram} to the optimal power flow (OPF) problem \citep{low_convex_2014}, as initially proposed by \citet{sondermeijer_regression-based_2016}.
We apply Thm. \ref{thm:restricted_communication} to determine the communication strategy that minimizes optimal distortion to further improve the reconstruction of the optimal actions $u_i^*$. 

Solving OPF requires a model of the electricity grid describing both topology and impedances; this is represented as a graph $\mathcal G = (\N,\E)$.
For clarity of exposition and without loss of generality, we introduce the linearized power flow equations over radial networks, also known as the \emph{LinDistFlow} equations \citep{baran_optimal_1989}:
\begin{IEEEeqnarray}{Rl}
P_{ij} =& \sum_{(j,k) \in \E, k \neq i} P_{jk} + p_{j}^\text{c} - p_{j}^\text{g} \,, \IEEEyesnumber\label{eq:BFeqs}\IEEEyessubnumber*\\
Q_{ij} =& \sum_{(j,k) \in \E, k \neq i} Q_{jk} + q_{j}^\text{c} - q_{j}^\text{g} \,,\\
v_{j} =& v_{i} - 2 \left( r_{ij} P_{ij} + \xi_{ij} Q_{ij} \right) \,
\end{IEEEeqnarray}
% Here, a radial network with balanced power flow is assumed, and power flow is described with the well studied DistFlow equations \eqref{eq:BFeqs}, first presented in \citep{baran_optimal_1989}. Consider a graph $\mathcal{G} (\mathcal{N},\mathcal{E})$ with $\mathcal{N}$ the set of nodes and $\mathcal{E}$ the set of edges representing the radial network. Furthermore, $\mathcal{C} \subset \mathcal{N}$ is the subset of $|\mathcal{C}| = I$ nodes that are equipped with controllable PV inverters.  
%A graphic representation of a line segment is shown in Figure \ref{fig:BFGraphic}. 
In this model, capitals $P_{ij}$ and $Q_{ij}$ represent real and reactive power flow on a branch from node $i$ to node $j$ for all branches $(i,j) \in \E$, lower case $p_i^c$ and $q_i^c$ are the real and reactive power consumption at node $i$, and $p_i^g$ and $q_i^g$ are its real and reactive power generation. 
Complex line impedances $r_{ij} + \sqrt{-1} \xi_{ij}$ have the same indexing as the power flows. 
The \emph{LinDistFlow} equations use the squared voltage magnitude $v_i$, defined and indexed at all nodes $i \in \N$.
These equations are included as constraints in the optimization problem to enforce that the solution adheres to laws of physics. 
%Note that for simplicity of presentation the DistFlow equations are formulated for a line feeder, however the method has a straightforward extension to networks with lateral branches.
% \begin{IEEEeqnarray}{Rl}
% P_{i+1} =& P_i - r_i \ell_i - p_{i+1}^\text{c} + p_{i+1}^\text{g} \IEEEyesnumber\label{eq:BFeqs}\IEEEyessubnumber*\\
% Q_{i+1} =& Q_i - x_i \ell_i - q_{i+1}^\text{c} + q_{i+1}^\text{g}\\
% v_{i+1} =& v_{i} - 2 \left( r_i P_i + x_i Q_i \right) + \left( r_i^2 + x_i^2 \right) \ell_i \\
% \ell_i =& \frac{P_i^2 + Q_i^2}{v_i} \IEEEyesnumber\label{eq:CurEq}
% \end{IEEEeqnarray}
% 
%\begin{figure}[!t]
%\centering
%\includegraphics[width=.35\textwidth]{PaperFigures/BranchFlowGraphic2.png}
%\caption{Schematic of a line segment including the symbols used in the DistFlow equations \eqref{eq:BFeqs}.}
%\label{fig:BFGraphic}
%\end{figure}
% 
% \subsubsection{Additional Constraints}
%Two additional constraints are included to account for inverter capacity and voltage goals. Also, a relaxation of the physical equality of \eqref{eq:CurEq} is introduced to relax the non-convex problem.
%

To formulate our decentralized learning problem, we will treat $x_i \triangleq (p_i^c,q_i^c,p^g_i)$ to be the local state variable, and, for all controllable nodes, i.e. agents $i \in \C$, we have $u_i \triangleq q_i^g$, i.e. the reactive power generation can be controlled ($v_i,P_{ij},Q_{ij}$ are treated as dummy variables). 
We assume that for all nodes $i \in \N$, consumption $p^c_i \ , \ q^c_i$ and real power generation $p^g_i$ are predetermined respectively by the demand and the power generated by a potential photovoltaic (PV) system.
The action space is constrained by the reactive power capacity $\left| u_i \right| = \left| q_i^\text{g} \right| \le \bar{q}_i$.
% at time $t$ of an inverter is limited by the total apparent power capacity $\bar{s}_i$ (constant) 
%which may vary in time due to real power generated by a local PV system.
% at time $t$. As such, the demand of reactive power does not interfere with real power generation.  The capacity constraint on each inverter can be formulated as
% \begin{equation}
% \left| q_i^\text{g}[t] \right|  \leq \bar{q}_i [t] = \sqrt{\bar{s}_i^2 - (p_i^\text{g}[t])^2} \,. \label{eq:InvCap}
% \end{equation}
%\begin{equation}
%\bar{q}^{\text{inverter}} = \sqrt{(s^{\text{inverter}})^2 - (p^{\text{inverter}})^2}, \label{eq:InvCap}
%\end{equation}
% Here, all inverters are assumed to have $5\%$ overcapacity, so $\bar{s} = 1.05 \bar{p}$, where $\bar{p}$ is the maximum generated real power of the PV system, allowing for available reactive capacity at all times. 
%Regulating the voltage in distribution circuits within the acceptable range specified by the American National Standard Institute (ANSI) Standard C84.1 for power quality is an essential responsibility for utility companies. 
In addition, voltages are maintained within $\pm 5\%$ of 120$V$, which is 
% as specified by ANSI Standard C84.1 
expressed as the constraint 
% \begin{equation}
$\underline{v} \leq v_{i} \leq \overline{v} \,$. \label{eq:VoltCons} 
% \end{equation} % Should we inlude graph notation?
% 
% The Second Order Cone Programming (SOCP) convex relaxation presented in \citep{farivar_inverter_2011} relaxes equality constraint \eqref{eq:CurEq} to inequality \eqref{eq:CurInEq}. 
%This relaxation is exact when over-satisfaction of loads is allowed \citep{Farivar2011}. Although it is allowed, over-satisfaction will not be observed in the solutions due to the non-increasing character of the loss minimizing part of the objective function \citep{Sojoudi2011}.
% \begin{equation}
% \ell_i \geq \frac{P_i^2 + Q_i^2}{v_i} \label{eq:CurInEq}
% \end{equation}
% 
% \subsubsection{Optimization Problem}
%In this work we consider the reactive power of inverters to be the only controllable variable (i.e. we do not alter real power consumption or generation). Perhaps you can already mention this earlier briefly in the intro and specifically
The OPF problem now reads
%This formulation of the optimal power flow problem grants flexibility in choosing an objective function, with the only requirements that the objective is convex and loss minimization is included.
%Although curtailment of PV generation can be included in this framework, we focus on inverter reactive power control. 
% Here, we will focus on minimizing voltage variability with respect to a reference voltage $v_{\text{ref}}$, which is caused by intermittency from renewable sources and the increasing load caused by electric vehicle charging. The associated optimal power flow problem is given by
%
\begin{align}
\label{eq:OPF}
u^* = \arg \min_{q_i^g \,, \ \forall i \in \C} \hspace{8pt} & \quad \sum_{i \in \N} | v_i - v_{\text{ref}} | \,,  \\
\text{s.t.} \hspace{8pt} & \quad \eqref{eq:BFeqs} \ , \ \left| q_i^\text{g} \right| \le \bar{q}_i \ , \ \underline{v} \leq v_{i} \leq \overline{v} \,. \nonumber
\end{align}
Following Fig. \ref{fig:flow_diagram}, we employ models of real electrical distribution grids (including the IEEE Test Feeders \citep{ieee_pes_ieee_2017}), which we equip with with $T$ historical readings $\{x[t]\}_{t=1}^{T}$ of load and PV data, which is composed with real smart meter measurements sourced from \citet{pecan_street_inc._dataport_2017}.
We solve \eqref{eq:OPF} for all data, yielding a set of minimizers $\{u^*[t]\}_{t=1}^{T}$. 
We then separate the overall data set into $C$ smaller data sets $\{x_i[t], u^*_i[t]\}_{t=1}^{T} \ , \ \forall i \in \mathcal C$ and train linear policies with feature kernels $\phi_i(\cdot)$ and parameters $\theta_i$ of the form $\hat \pi_i (x_i) = \theta_i^{\top} \phi_i(x_i)$.
%We now consider the problem of parameterizing the policies $\hat \pi_i (x_i; \theta_i, \phi_i(\cdot))$, where the data $x_i$ can be transformed by a kernel function $\phi_i(\cdot)$ that is potentially different for each node, and $\theta_i$ denotes the regression parameters. If we restrict the regression function to be a linear combination of a vector of features $\phi_i(x_i)$, this yields
%\begin{equation}
%\hat \pi_i (x_i) = \theta_i^{\top} \phi_i(x_i) \ , \quad \forall i \in \mathcal C \,.
%\end{equation}
Practically, the challenge is to select the best feature kernel $\phi_i(\cdot)$. We extend earlier work which showed that decentralized learning for OPF can be done satisfactorily via a hybrid forward- and backward-stepwise selection algorithm \citep[Chapter 3]{friedman_elements_2001} that uses a quadratic feature kernels.

Figure \ref{fig:result_voltages} shows the result for an electric distribution grid model based on a real network from Arizona. 
This network has 129 nodes and, in simulation, 53 nodes were equipped with a controllable DER (i.e. $N = 129, C = 53$). 
%The results empirically verify that regression-based decentralization can mimic the optimal centralized policy found by solving OPF \eqref{eq:OPF}, via an implementation of regression-based policies that \emph{solely} rely on local measurements. 
In Fig. \ref{fig:result_voltages} we show the voltage deviation from a normalized setpoint on a simulated network with data not used during training. The improvement over the no-control baseline is striking, and performance is nearly identical to the optimum achieved by the centralized solution. Concretely, we observed: (i) no constraint violations, and (ii) a suboptimality deviation of 0.15\% on average, with a maximum deviation of 1.6\%, as compared to the optimal policy $\pi^*$.

In addition, we applied Thm. \ref{thm:restricted_communication} to the OPF problem for a smaller network \citep{ieee_pes_ieee_2017}, in order to determine the optimal communication strategy to minimize a squared error distortion measure. 
Fig. \ref{fig:comparing_communication_strategies_opf} shows the mean squared error distortion measure for an increasing number of observed nodes $k$ and shows how the optimal strategy outperforms an average over random strategies. 

\begin{figure}[t]
\centering
\begin{subfigure}{.5\textwidth}
  \centering
  \includegraphics[height=0.18\textheight]{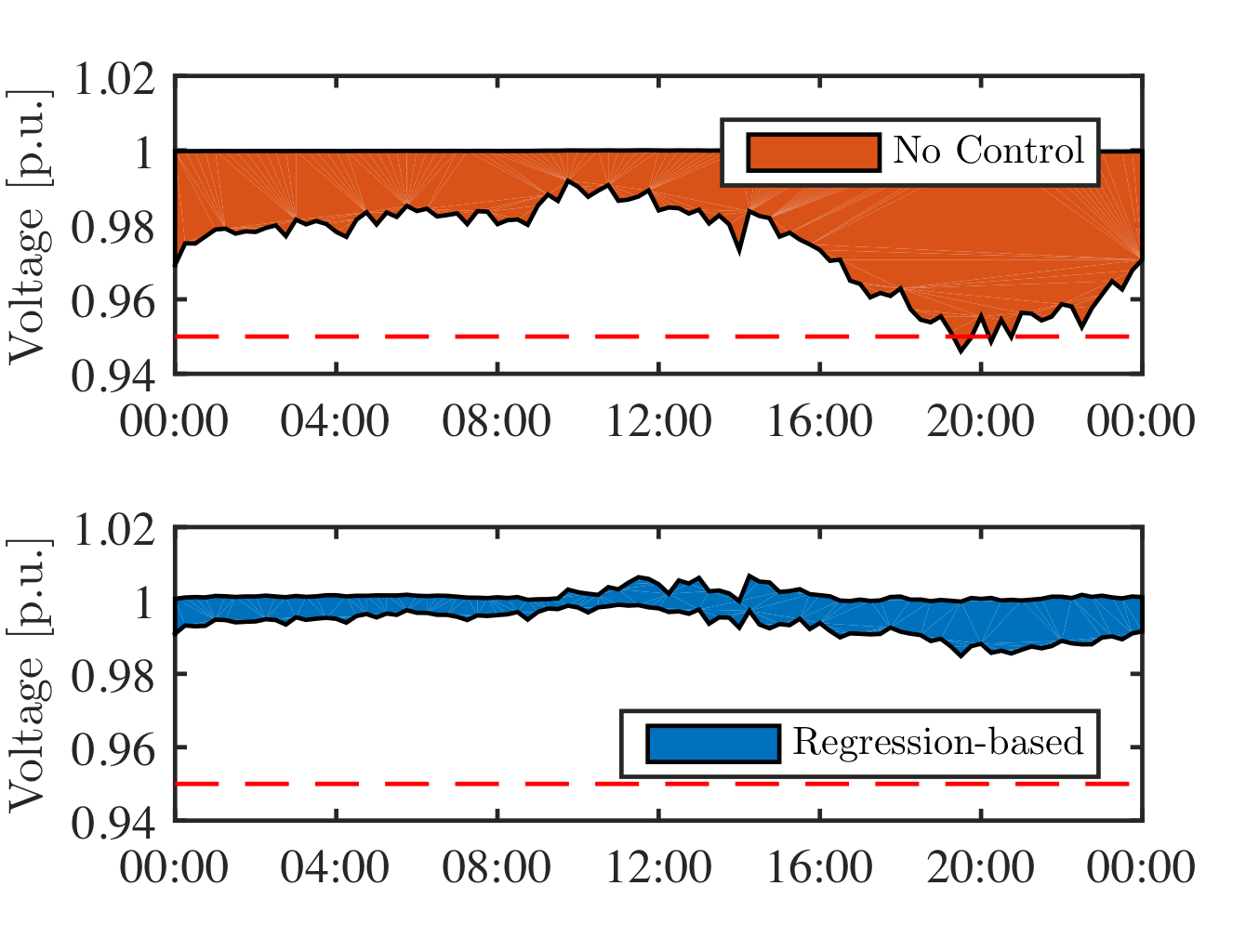}
  \caption{Voltage output with and without control. \label{fig:result_voltages}}
\end{subfigure}%
\begin{subfigure}{.5\textwidth}
  \centering
  \includegraphics[height=0.18\textheight]{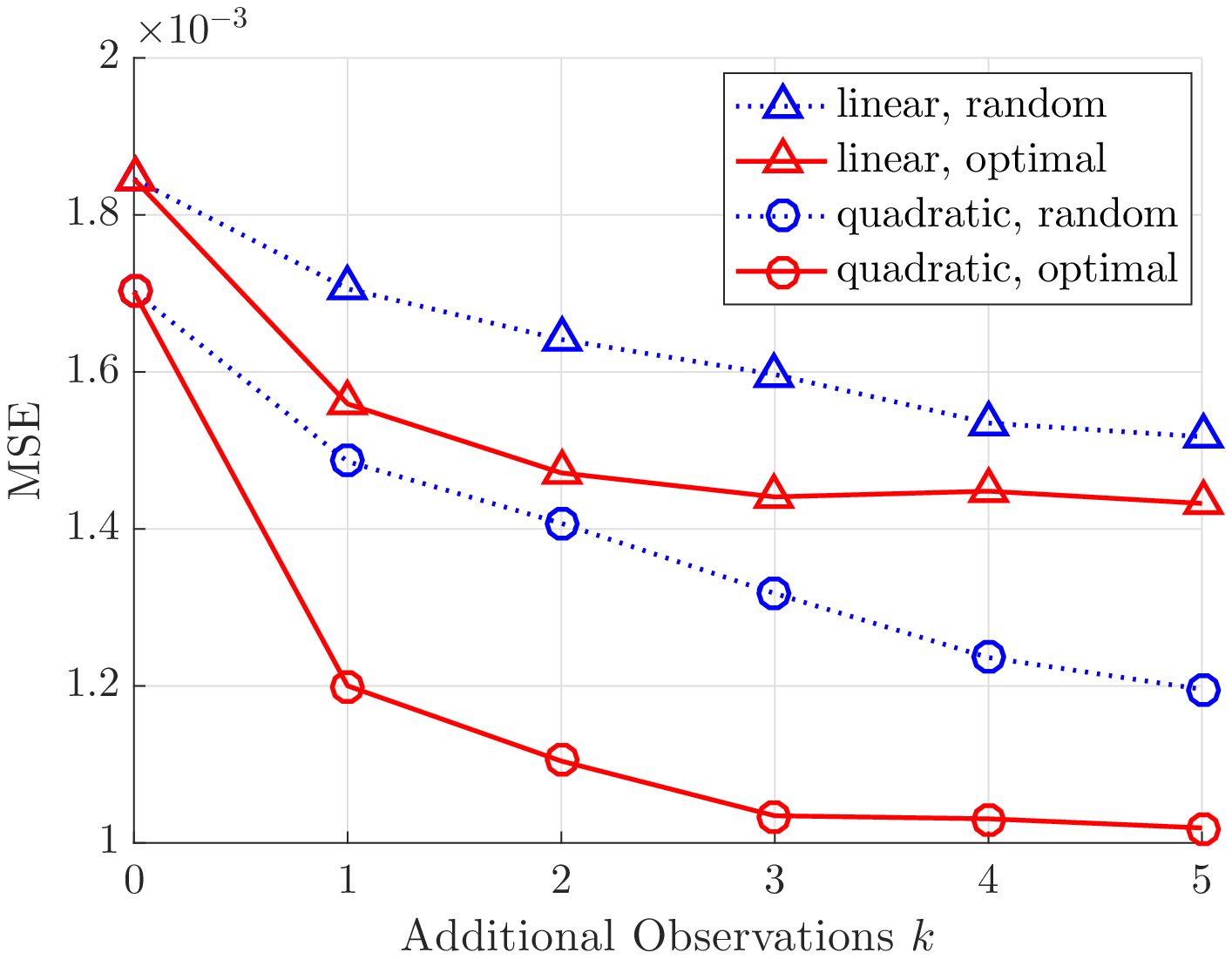}
  \caption{Comparison of OPF communication strategies. \label{fig:comparing_communication_strategies_opf}}
\end{subfigure}%
\caption{Results for decentralized learning on an OPF problem. (a) shows an example result of decentralized learning - the shaded region represents the range of all voltages in a network over a full day. As compared to no control, the fully decentralized regression-based control reduces voltage variation and prevents constraint violation (dashed line). (b) shows that the optimal communication strategy $\mathcal{S}_i$ outperforms the average for random strategies on the mean squared error distortion metric. The regressors used are stepwise linear policies~$\hat \pi_i$ with linear or quadratic features.  \label{fig:results_opf}}
\end{figure}

\section{Conclusions and Future Work}
\label{sec:discussion_future_work}

This paper generalizes the approach of \citet{sondermeijer_regression-based_2016} to solve multi-agent static optimal control problems with decentralized policies that are learned offline from historical data.
Our rate distortion framework facilitates a principled analysis of the performance of such decentralized policies and the design of optimal communication strategies to improve individual policies.
These techniques work well on a model of a sophisticated real-world OPF example.

There are still many open questions about regression-based decentralization. It is well known that strong interactions between different subsystems may lead to instability and suboptimality in decentralized control problems~\citep{davison_decentralized_1990}. There are natural extensions of our work to address dynamic control problems more explicitly, and stability analysis is a topic of ongoing work. Also, analysis of the suboptimality of regression-based decentralization should be possible within our rate distortion framework.
% We observed such instability phenomena when using the voltage $v_i$ as a local state, which led to the insight that the local power states are weakly coupled and a safe choice for decentralized learning (and control). 
%We are studying stability properties in the spatial dynamics of electrical grids, and plan to compare the approach with the literature on stability for general nonlinear dynamics and control~\citep{sastry_nonlinear_2013}.
% 
% While we can provide a lower bound on average distortion, we \textit{do not} provide a means of achieving this bound in all cases. The synthetic Gaussian example is one instance in which the bound is met; the OPF example however does not meet the bound. In general, we are left with the difficult problem of parameterizing the arbitrary distribution from which network state $x$ and optimal action $u^*$ are drawn, or equivalently with choosing the optimal parametric form of a regression policy $\hat \pi_i$. This is a very difficult problem which we leave for future work.
% 
Finally, it is worth investigating the use of deep neural networks to parameterize both the distribution $p(u^* | x)$ and local policies $\hat \pi_i$ in more complicated decentralized control problems with arbitrary distortion measures.

\bibliographystyle{abbrvnat}
\bibliography{references_roel}

\end{document}